%% file: main.tex
\documentclass[11pt]{article}
\usepackage[T1]{fontenc}
\usepackage[utf8]{inputenc}
\usepackage{diagbox}
\usepackage{fullpage}
\def\withcolors{0}
\def\withnotes{0}

\usepackage{tikz}
\usetikzlibrary{arrows}
\usetikzlibrary{calc,decorations.pathmorphing,patterns, fit}

\PassOptionsToPackage{algo2e,ruled,linesnumbered}{algorithm2e}
\RequirePackage{algorithm2e}
\usepackage{algorithm}
\usepackage{algorithm2e}

\ifnum\withcolors=1
\else

\fi

\usepackage[colorinlistoftodos,textsize=scriptsize]{todonotes}
\ifnum\withnotes=1
  \newcommand{\anote}[1]{\par\acolor{\textbf{A:} #1}} %
    \newcommand{\scolor}[1]{{\color{orange}#1}} %
  \newcommand{\snote}[1]{\par\scolor{\textbf{S:} #1}} %

	\newcommand{\marginnote}[1]{\todo[color=white,linecolor=black]{#1}}
\else
  
  \newcommand{\anote}[1]{{#1}}
  \newcommand{\scolor}[1]{{#1}}
  \newcommand{\snote}[1]{{#1}}
  
  \newcommand{\marginnote}[1]{\ignore{#1}}
\fi
\newcommand{\amargin}[1]{\marginnote{\anote{#1}}}
\newcommand{\smargin}[1]{\marginnote{\snote{#1}}}

\input{glodef}

\input{locdef}

\newcommand{\newzs}[1]{{\textcolor{black}{#1}}}
\title{Estimating Sparse Discrete Distributions Under Local Privacy and Communication Constraints}
\author{Jayadev Acharya\\
	Cornell University\\
	\tt{acharya@cornell.edu}
	\and
	Peter Kairouz\\
	Google\\
	\tt{kairouz@google.com}
	\and
	Yuhan Liu\\
	Cornell University\\
	\tt{yl2976@cornell.edu}
	\and
	Ziteng Sun\\
	Cornell University\\
	\tt{zs335@cornell.edu}
}
\date{}
\begin{document}

\maketitle
\begin{abstract}
	We consider the problem of estimating sparse discrete distributions 
	under 
	local differential privacy (LDP) and communication constraints. We 
	characterize the sample complexity for sparse estimation under LDP 
	constraints up to a constant factor, and the sample complexity under
	communication constraints up to a logarithmic factor. Our upper bounds 
	under LDP are based on the Hadamard Response, a private coin scheme 
	that 
	requires only one bit of communication per user. Under communication 
	constraints we propose public coin schemes based on random hashing 
	functions. Our tight lower bounds are based on recently proposed 
	method 
	of chi squared contractions.
\end{abstract}
\section{Introduction}
\label{sec:introduction}
\input{introduction}

\subsection{Notations and problem set-up}\label{sec:prelim}
\input{preliminary}

\subsection{Previous results and our contribution}
\label{sec:results}
\input{results}

\subsection{Related work}
\label{sec:related}
\input{related}

\section{Sparse estimation under LDP constraints}
\label{sec:ldp}
\input{upper_LDP}

\input{lower_LDP}

\section{Sparse estimation under communication constraints}
\label{sec:communication}
\input{upper_communication_new}

\input{lower_communication}

\section{Experiments}
\label{sec:experiment}
\input{experiments}

\bibliographystyle{plain}
\bibliography{masterref}

\appendix
\section{Proof of Lemma~\ref{lem:sparse}}
\label{app:sparse-ub}
\input{app-lemma-covering.tex}
\section{Proof of Lemma~\ref{lem:step-one} and Lemma~\ref{thm:utility_com}}
\label{app:sparse-comm}
\input{app-comm-lemmas.tex}

\input{appendix}

\end{document}

%% file: glodef.tex
\usepackage{lmodern}
\usepackage{xspace}
\usepackage{amsfonts,amsmath,amssymb,amsthm, pbox}

\usepackage[colorlinks,citecolor=blue,bookmarks=true]{hyperref}

\usepackage{multirow}

\usepackage{dsfont} %

\usepackage[shortlabels]{enumitem}
\setitemize{noitemsep,topsep=0pt,parsep=0pt,partopsep=0pt}
\setenumerate{noitemsep,topsep=0pt,parsep=0pt,partopsep=0pt}

\makeatletter
\@ifundefined{theorem}{%
  \theoremstyle{plain}
  \newtheorem{theorem}{Theorem}

  \newtheorem{lemma}[theorem]{Lemma}
  \newtheorem{claim}{Claim}

  \theoremstyle{remark}

}{}
\makeatother

\newcommand{\ignore}[1]{}

\newcommand{\II}{\mathbb{I}} %
\newcommand{\EE}{\mathbb{E}}

\newcommand{\RR}{\mathbb{R}}

\newcommand{\expectation}[1]{\EE\left[#1\right]}
\newcommand{\variance}[1]{\textrm{Var}\left(#1\right)}

\def \cC     {{\cal C}}

\def \cP     {{\cal P}}

\def \cW     {{\cal W}}
\def \cX     {{\cal X}}
\def \cY     {{\cal Y}}
\def \cZ     {{\cal Z}}

\def \Paren#1{{\left({#1}\right)}}

\newcommand{\eqdef}{{:=}}

\newcommand{\probof}[1]{\Pr\Paren{#1}}

\def\ignore#1{}

\newcommand{\bi}{\begin{itemize}}
\newcommand{\ei}{\end{itemize}}

\def\orpro{\mathop{\mathchoice
   {\vee\kern-.49em\raise.7ex\hbox{$\cdot$}\kern.4em}
   {\vee\kern-.45em\raise.63ex\hbox{$\cdot$}\kern.2em}
   {\vee\kern-.4em\raise.3ex\hbox{$\cdot$}\kern.1em}
   {\vee\kern-.35em\raise2.2ex\hbox{$\cdot$}\kern.1em}}\limits}

\def\andpro{\mathop{\mathchoice
 {\wedge\kern-.46em\lower.69ex\hbox{$\cdot$}\kern.3em}
 {\wedge\kern-.46em\lower.58ex\hbox{$\cdot$}\kern.25em}
 {\wedge\kern-.38em\lower.5ex\hbox{$\cdot$}\kern.1em}
 {\wedge\kern-.3em\lower.5ex\hbox{$\cdot$}\kern.1em}}\limits}

\def\simge{\mathrel{%
   \rlap{\raise 0.511ex \hbox{$>$}}{\lower 0.511ex \hbox{$\sim$}}}}

\def\simle{\mathrel{
   \rlap{\raise 0.511ex \hbox{$<$}}{\lower 0.511ex \hbox{$\sim$}}}}

%% file: locdef.tex
\newcommand{\err}{\xi}

\newcommand{\Dk}{\Delta_k}

\newcommand{\ab}{k}
\newcommand{\oab}{K}
\newcommand{\ns}{n}
\newcommand{\nse}{{\frac{n}{2}}}
\newcommand{\nsein}{n/2}
\newcommand{\ms}{M}

\newcommand{\p}{p}

\newcommand{\bias}{b}

\newcommand{\hp}{\hat{\p}}
\newcommand{\tp}{\tilde{\p}}

\newcommand{\eps}{\varepsilon}

\newcommand{\ham}[2]{d_{\rm Ham}(#1,#2)}
\newcommand{\dtv}[2]{d_{\rm TV}(#1,#2)}
\newcommand{\tvd}{d_{\rm TV}}

\newcommand{\norm}[2]{\left\lVert#1\right\rVert_{#2}}

\newcommand{\setk}{\triangle_k}
\newcommand{\setsp}{\triangle_{k,s}}
\newcommand{\SC}{\text{SC}}

\newcommand{\expectover}[2]{\EE_{#1}\left[{#2}\right]}
\newcommand{\condexpect}[2]{\expectation{{#1}\;\middle|\;{#2}}}
\newcommand{\condprob}[2]{\Pr\left({#1}\;\middle|\;{#2}\right)}

\newcommand{\pset}{t}
\newcommand{\psetemp}{\hat{\pset}}
\newcommand{\vecpset}{\mathbf{\pset}}
\newcommand{\vecpsetemp}{\hat{\vecpset}}

\newcommand{\citep}[1]{\cite{#1}}

%% file: introduction.tex
Estimating distributions from data samples is a central task in statistical 
inference. In modern learning systems such as federated 
learning~\citep{kairouz2019advances},
data is generated from distributed 
sources including cell phones, wireless sensors, and smart healthcare 
devices. Access to such data is subject to severe ``local information 
constraints'', such as communication and energy constraints, privacy 
concerns. For several statistical inference tasks, including distribution 
estimation, privacy and communication constraints lead to significant 
degradation in utility (see Section~\ref{sec:related} for a detailed 
discussion). Moreover, in some applications, such as web-browsing, 
genomics, and 
language modeling, the distribution is often supported over a small 
unknown subset of the domain. Motivated by the utility gain in high 
dimensional statistics under sparsity 
assumptions~\citep{wainwright2019high}, we study the problem of 
estimating sparse discrete distributions under privacy and communication 
constraints. 

%% file: preliminary.tex
Let $[\ab] := \{1, 2, ..., \ab \}$ and $\triangle_\ab := \{ p \in [0,1]^\ab  \colon  \sum_{x\in [\ab]} p(x)= 1 \}$ be the set of all distributions over $[\ab]$. {For $\p \in \Dk$ and $S \subset [k]$, let $\p^{S}$ be the vector restricted on indices in $S$.} Independent samples $X_1, \ldots, X_\ns$ from an unknown $p\in\triangle_\ab$ are observed by $\ns$ users, where user $i$ observes $X_i$. User $i$ sends a message $Y_i = W_i(X_i)$ to a central server, where $W_i: [\ab] \rightarrow \cY$ is a randomized mapping (channel) with 

\[
    W_i(y \mid x) = \condprob{Y_i = y}{X_i = x}.
\]

\noindent We consider privacy and communication constrained messages in this paper, which can be enforced by restricting $W_i$s to belong to a class $\cW$ of \emph{allowed} channels. 

\medskip
\noindent\textbf{Local Differential Privacy (LDP).} A channel $W:[\ab]\to\cY=\{0,1\}^\ast$ is $\eps$-LDP if
\begin{align}
    \sup_{y \in \cY} \sup_{x, x' \in \cX} \frac{W(y \mid x)}{W(y \mid x')} \le e^{\eps}.\label{eqn:lsp}
\end{align}
$\cW_{\eps}=\{W: W \text{ is } \eps\text{-LDP}\}$ is the set of all 
$\eps$-LDP channels.

\medskip
\noindent\textbf{Communication constraints.} Let $\ell<\log \ab$, and  $\cW_\ell\:=\{W:[\ab]\to\cY=\{0,1\}^\ell\}$ be the set of channels that output $\ell$-bit messages, and thus characterize communication constraints.

\begin{figure}[h]\centering
	\scalebox{.75}{\input{fig-ic}}
	\caption{Distributed inference with simultaneous message passing (SMP) protocol.}
	\label{fig:model}
\end{figure}
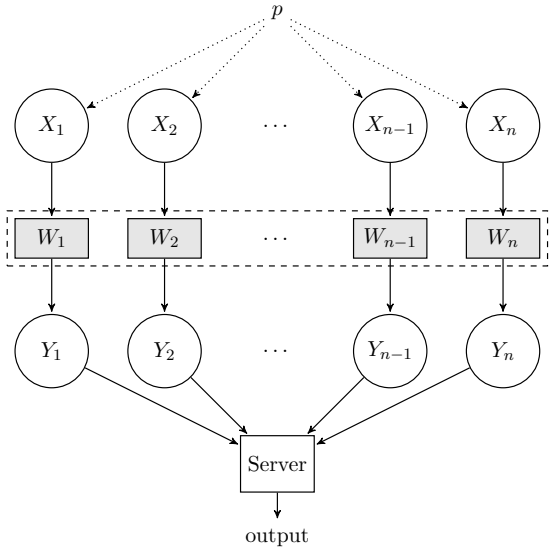
\medskip
\noindent\textbf{Distribution estimation.} Let 
$d:\triangle_\ab\times\triangle_\ab\to\RR_+$ be a distance measure. A 
distribution estimation protocol under constraints $\cW$ is a set of 
channels $W^\ns := (W_1, W_2, \ldots, W_\ns) \in \cW^\ns$ and an 
estimator $\hat{p}: \cY^\ns \rightarrow \triangle_\ab$. Upon observing the 
$n$ output messges $Y^\ns:=(Y_1, Y_2, \ldots, Y_\ns)$, the central server 
outputs an estimate $\hat{p}(Y^\ns)$ of the underlying unknown 
distribution $p$.  Let $\alpha>0$ be an accuracy parameter. 
The minimax sample complexity for estimation is
\[
	\SC(\alpha, \triangle_\ab, \cW, d) := \arg\min_{\ns} \left \{ \min_{\hat{p}} \min_{W^\ns \in \cW^\ns} \max_{p \in \triangle_\ab}   \probof{d(p, \hat{p}(Y^\ns)) \le \alpha} \ge 0.9 \right  \},
\]
the fewest number of samples for which we can estimate every $p \in 
\setk$ up to $\alpha$ accuracy with probability at least 0.9. 
We will use the total variation distance, $\dtv{\hat{p}}{p} := 
\sup_{S\subseteq [k]}|\hat{p}(S) - {p}(S)| = \frac{1}{2}\norm{p - \hat{p}}{1}$ 
as the distance measure. 

\medskip
\noindent\textbf{Sparse distribution estimation.} Let $s\le \ab/100$\footnote{Sparsity larger than $k/100$ gives same answers as the non-sparse case.} and 
\[
\setsp := \{ p \in \triangle_k \colon \norm{p}{0} \le s\}
\]
be the distributions in $\triangle_\ab$ with support size at most $s$. 
Let 
\[
	\SC(\alpha, \setsp, \cW, \tvd) \eqdef \arg\min_{n} \left \{ \min_{\hat{p}} \min_{W^n \in \cW^n} \max_{p \in \setsp}   \probof{\dtv{p}{\hat{p}(Y^n)}) \le \alpha} \ge  0.9 \right \}
\]
denote the sample complexity of estimating $s$-sparse distributions to total variation distance $\alpha$.

In this paper, we consider simultaneous message passing (SMP) protocols 
(non-interactive schemes) where all the messages from users are sent 
simultaneously (see Figure~\ref{fig:model}). SMP protocols are broadly 
classified as \emph{private-coin} and \emph{public-coin} protocols. In 
{private-coin} schemes the channels $W_i$ are independent. In the more 
general {public-coin} schemes,  the channels are chosen based on a 
function of a public randomness $U$ observed by all the users  and the 
server. Private-coin protocols are a strict subset of public-coin protocols. 
We refer the readers to~\cite{acharya2020interactive} for detailed 
definitions of these protocols.  %

%% file: fig-ic.tex
\begin{tikzpicture}[->,>=stealth',shorten >=1pt,auto,node distance=20mm, semithick]
  \node[circle,draw,minimum size=13mm] (A)
  {$X_1$}; \node[circle,draw,minimum size=13mm] (B) [right of=A]
  {$X_2$}; \node (C) [right of=B] {$\dots$}; \node[circle,draw,minimum
  size=13mm] (D) [right of=C] {$X_{\ns-1}$}; \node[circle,draw,minimum
  size=13mm] (E) [right of=D] {$X_\ns$};
  
  \node[rectangle,draw,minimum width=13mm,minimum
  height=7mm,fill=gray!20!white] (WA) [below of=A]
  {$W_1$}; \node[rectangle,draw,minimum width=13mm,minimum
  height=7mm,fill=gray!20!white] (WB) [below of=B] {$W_2$}; \node (WC)
  [below of=C] {$\dots$}; \node[rectangle,draw,minimum
  width=13mm,minimum height=7mm,fill=gray!20!white] (WD) [below of=D]
  {$W_{\ns-1}$}; \node[rectangle,draw,minimum width=13mm,minimum
  height=7mm,fill=gray!20!white] (WE) [below of=E] {$W_\ns$};
  
  \node[draw,dashed,fit=(WA) (WB) (WC) (WD) (WE)] {};
  o \node[circle,draw,minimum size=13mm] (YA) [below of=WA]
  {$Y_1$}; \node[circle,draw,minimum size=13mm] (YB) [below of=WB]
  {$Y_2$}; \node (YC) [below of=WC]
  {$\dots$}; \node[circle,draw,minimum size=13mm] (YD) [below of=WD]
  {$Y_{\ns-1}$}; \node[circle,draw,minimum size=13mm] (YE) [below
  of=WE] {$Y_\ns$};

  \node (P) [above of=C] {$\p$}; \node[rectangle,draw, minimum
  size=10mm] (R) [below of=YC] {Server}; \node (out) [below
  of=R,node distance=13mm] {output};

  \draw[->] (P) edge[dotted] (A)(A) edge (WA)(WA) edge (YA)(YA) edge
  (R); \draw[->] (P) edge[dotted] (B)(B) edge (WB)(WB) edge (YB)(YB)
  edge (R); \draw[->] (P) edge[dotted] (D)(D) edge (WD)(WD) edge
  (YD)(YD) edge (R); \draw[->] (P) edge[dotted] (E)(E) edge (WE)(WE)
  edge (YE)(YE) edge (R); \draw[->] (R) edge (out);
\end{tikzpicture}

%% file: results.tex
Discrete distribution estimation under communication~\citep{HOW:18, 
HMOW:18, AcharyaCT19, acharya2020interactive} and 
LDP~\citep{DuchiJW13, ErlingssonPK14, KairouzBR16, YeB18, AcharyaSZ18, 
AcharyaS19} constraints is well studied, and it is now known that for 
$\ell$-bit channels $\cW_\ell$, and  $\eps$-LDP 
channels $\cW_\eps$ (for $\eps=O(1)$),
\begin{align}
   \SC(\alpha, \triangle_\ab, \cW_{\ell}, \tvd) = \Theta\Paren{\frac{k^2}{\alpha^2 \min \{k, 2^{\ell} \}}},\quad   \SC(\alpha, \triangle_k, \cW_{\eps},  \tvd) = \Theta\Paren{\frac{k^2}{\alpha^2 \eps^2}}.
   \label{eqn:sample-complexity}
\end{align}
Plugging $\ell = \log \ab$ in the first equation gives the centralized sample complexity of $\Theta\Paren{{\ab}/{\alpha^2}}$.  Note that for $\eps=O(1)$ and $\ell=O(1)$, the sample complexity increases by a factor $k$ from the centralized setting. 

We now present our results on estimating distributions in $\triangle_{\ab,s}$, the set of $s$ sparse distributions in $\Dk$. 
Our first result is a complete characterization of the sample complexity  under $\eps$-LDP  up to constant factors. 
\begin{theorem} [Bounds for $\eps$-LDP constraint]
\label{thm:ldp_sparse}
    For $\eps = O(1)$ and $\alpha \in (0,1)$,
    \[
        \SC(\alpha, \setsp, \cW_{\eps}, \tvd) = \Theta\Paren{\frac{s^2 \max \{ \log (k/s), 1\} }{\alpha^2 \eps^2}}.
    \]
    Moreover, there exists a private-coin protocol with one-bit  privatized messages that achieves the upper bound. The algorithm runs in nearly linear time in $n$ and $\ab$.
\end{theorem}
A few remarks are in order. This result shows that sparsity $s$ is the 
effective domain size up to logarithmic factors. While an additional $\log 
\ab$ factor is slightly simpler to obtain, a more involved technique in 
sparse estimation based on a covering set argument is used to establish 
the upper bound with the optimal overhead factor of $\log (\ab/s)$. The 
lower bound is obtained by applying the recently developed chi-squared 
contraction techniques~\citep{AcharyaCT19} to a new construction of 
distributions.

\ignore{\noindent It can be seen that the sample complexity is almost the sample as the sample complexity of discrete distribution over $[s]$ up to a logarithmic factor in the true dimension $k$, which is a significant improvement when $s \ll k$. The algorithm that achieves Algorithm~\ref{thm:ldp_sparse} is aware of the sparsity parameter $s$. However, in certain applications, the analyst doesn't know how sparse the distribution is. In this case, we can obtain a similar improvement, stated in Theorem~\ref{thm:ldp_sparse_b}.
\begin{theorem} \label{thm:ldp_sparse_b}
	There exists an $\eps-$LDP algorithm, which doesn't know $s$, that achieves a sample complexity of
	\[
	O\left(\frac{s^2\log \ab}{\alpha^2\eps^2}\right).
	\]
	Moreover, the algorithm is a private-coin protocol with 1-bit communication from each user achieving this bound. The algorithm runs in quasilinear time.
\end{theorem}
\noindent Algorithms achieving Theorem~\ref{thm:ldp_sparse} and Theorem~\ref{thm:ldp_sparse_b} are presented in Section~\ref{sec:ldp}. The optimality of the sample complexity bound in Theorem~\ref{thm:ldp_sparse} is presented in Section~\ref{sec:ldp_lower}.}
\medskip

\noindent We now present our sample complexity bounds under communication constraints. Unlike LDP, our bounds are off by logarithmic factors in various parameter regimes. Resolving this gap and obtaining the tight bounds is an open question.
\begin{theorem}[Bounds for $\ell$-bit constraint]
\label{thm:lbit_sparse}
For $\alpha\in(0, 1)$, and channel family $\cW_\ell$, the sample complexity for learning distributions in $\setsp$ satisfies,
\[
    \SC(\alpha, \setsp, \cW_\ell,  \tvd)=O\left(\frac{s^2 \max \{ \log({\ab}/{s}),1\}}{\alpha^2\min\left\{2^\ell, s\right\}}\right).
\]
\[
   \SC(\alpha, \setsp, \cW_\ell,  \tvd)=\Omega\left(\max\left\{\frac{s^2}{\alpha^2 \min\{2^\ell, s\}}, \frac{s^2\max\{\log({\ab}/{s}),1\}}{\alpha 2^\ell}\right\}\right).
\]
\end{theorem}

\noindent We briefly discuss the lower bound. Despite the gap, it is optimal for constant $\alpha$. The first term is exactly the lower bound when the support is known. The second term dominates when $\log (k/s)>1/\alpha$, i.e., the support is very sparse. 

\medskip
\noindent \textbf{Organization.} The remainder of the paper is organized as follows. We discuss related works in Section~\ref{sec:related}. We present algorithms and proofs for sparse estimation under LDP and communication constraints in Section~\ref{sec:ldp} and Section~\ref{sec:communication} respectively.

%% file: related.tex
Distribution estimation has a rich literature (see e.g., \cite{BarlowBBB72, Silverman86, DevroyeG85, DevroyeL01}, and references therein). There has been recent interest in distributed distribution estimation under communication and privacy constraints. For estimating discrete distributions under communication constraints, the optimal sample complexity is established in~\cite{HOW:18, pmlr-v97-acharya19a, barnes2019learning}, and for local privacy constraints in~\cite{DuchiJW13, ErlingssonPK14, KairouzBR16, YeB18, Bassily18, AcharyaSZ18, AcharyaS19}. \cite{AcharyaCT19, acharya2020interactive} unify both constraints under the framework of distributed inference under local information constraints, where optimal bounds are obtained under both non-interactive and interactive protocols. \cite{chen2020breaking} considers the trade-off between privacy and communication constraints and provides optimal bounds in all parameter regimes. \cite{murakami2019utility, acharya2020contextaware} study discrete distribution estimation under different privacy constraints on the symbols. 

\medskip
\noindent \cite{KairouzBR16} proposed an LDP distribution estimation 
algorithm for ``open alphabets'' which applies to the sparse setting. 
However, it requires public randomness and to the best of our knowledge, 
no theoretical guarantee of this 
method is provided for sparse distribution estimation. In contrast, the our 
algorithm is a private-coin algorithm that only requires one bit per user, 
and its sample complexity is proven to be optimal up to constant factors. 

\medskip
\noindent A closely related problem is heavy hitter detection under LDP 
constraints~{\citep{BassilyS15, JMLR:v21:18-786}, where no distributional 
assumption on the data is made. A modification of their heavy hitters 
algorithms provides a sub-optimal $O(s^2\log k/\alpha^2 \eps^2)$ 
sample complexity in terms of $\ell_1$ error for $\eps$-LDP distribution estimation .} 

\medskip
\noindent Statistical inference with sparsity assumption has been studied 
extensively for decades. The closest to our works are the Gaussian 
sequence model and high dimensional linear 
regression~\citep{donoho1994minimax, raskutti2011minimax, 
duchi2013distance}. In these applications, it is assumed that the 
observations are linear transforms of the underlying parameter plus 
independent Gaussian noises on each dimension. In 
Section~\ref{sec:ldp_upper_a}, it can be seen that using 
Algorithm~\ref{alg:hadamard}, the histogram of observations can also be 
seen as a linear transform of the parameter of interest, however, with 
dependent noises on each dimension. We borrow ideas from these works in 
proving the upper bound. However, the lower bound part requires new 
proofs due to the dependency structure.

\medskip 
\noindent A few recent works study sparse estimation under information constraints. \cite{DuchiJW13} and~\cite{wang2019sparse} consider the 1-sparse case and study mean estimation and linear regression under LDP constraints respectively. \scolor{\cite{pmlr-v99-duchi19a} provides lower bounds for sparse Gaussian mean estimation under LDP constraints via communication complexity.} \cite{barnes2020fisher} considers estimating the mean of product Bernoulli distribution when the mean vector is sparse, which is different from the $k$-ary setting considered in this paper.  \scolor{\cite{shamir2014fundamental} considers the problem of detecting the biased coordinate of product Bernoulli distributions under communication constraints, which can be viewed as a 1-sparse detection problem.}
\cite{ZhangDJW13, garg2014communication, BravermanGMNW16, HOW:18} consider sparse Gaussian mean estimation under communication constraints (\cite{garg2014communication, BravermanGMNW16} consider interactive protocols with the goal of bounding the total amount of communication from all users). It was shown that under a fixed communication budget, the rate still scales linearly with the ambient dimension of the problem instead of the logarithmic dependence in the discrete case considered in this paper.

%% file: upper_LDP.tex
In this section we will establish the sample complexity of sparse 
distribution estimation under LDP constraints. 
In Section~\ref{sec:ldp_upper_a}, we analyze a private-coin algorithm 
where each user sends only one-bit messages detailed in 
Algorithm~\ref{alg:hadamard}. The algorithm has two 
steps, listed below.
\begin{enumerate}
\item 
Using the private-coin Hadamard Response algorithm in~\cite{AcharyaS19}, players send one-bit messages. 
\item
The server projects a vector obtained from these messages onto $\setsp$ to obtain the final estimate.
\end{enumerate}
 We note that this algorithm is similar to that in~\cite{AcharyaS19, 
 Bassily18} where in the projection step they project onto $\triangle_{\ab}$ 
 to estimate distributions without sparsity assumptions. While the 
 algorithm is simple, to obtain the tight upper bounds, our analysis relies 
 on a standard but involved covering-based techniques in sparse 
 estimation. We also remark that a sample-optimal scheme can also be 
 obtained using the popular RAPPOR mechanism~\citep{ErlingssonPK14, 
 KairouzBR16}, which has higher communication overhead. We present this 
 algorithm in the appendix for completeness.

In Section~\ref{sec:ldp_lower}, we present a matching lower bound to prove 
the optimality of the aforementioned algorithm. The proof relies on 
applying the recently developed chi-squared contraction 
method~\citep{AcharyaCT19} and a variant of Fano's inequality 
in~\cite{duchi2013distance}.
\amargin{do we still use this variant of duchi?}
\smargin{YES}

\begin{algorithm}[h]
	\caption{1-bit Hadamard Response with Projection}
	\label{alg:hadamard}
	\LinesNumbered
	\KwIn{$X_1, \ldots, X_n$ i.i.d. from $p\in\setsp$, the sparsity parameter $s$. }
	\KwOut{$\hp\in\triangle_\ab:$ an estimate of $p$.} 
	
	Let $\oab=2^{\lceil\log_2(k+1) \rceil}$ be the smallest power of 2 more 
	than $\ab$.\label{step:one}
	
	For $y \in[\oab]$, let $B_y:=\{x\in[\oab]:H_{\oab}(x, y)=1\}$ be 
	the rows where the $y$th column has 1. 
	
	Divide the $\ns$ users into $\oab$ sets $S_1, \ldots, S_\oab$ deterministically by assigning all $i \equiv  j \mod K$ to $S_j$ for $i\in[\ns]$.
	
	$\forall j \in [\oab]$ and $\forall i \in S_j$, the distribution of the 
	one-bit message $Y_i$ is
	\begin{align}
		\label{equ:scheme}
		\Pr(Y_i=1)=\begin{cases}
			\frac{e^\eps}{e^\eps + 1}, &X_i\in B_{j},\\
			\frac{1}{e^\eps + 1},&\text{otherwise},
		\end{cases}
	\end{align}
	namely if $H_\oab(X_i, j) = 1$, we send 1 with higher probability than 0.  
	
	Let $\vecpsetemp:=(\psetemp_1, \ldots, \psetemp_{\oab})$ where $\forall j \in [\oab]$,
$\psetemp_j :=\frac{1}{|S_j|}\sum_{i\in S_j}Y_i$ is the fraction of messages from $S_j$ that are 1.  
	
	Compute intermediate estimates for 
	\[ 
		\tp_\oab:=\frac{e^\eps+1}{\oab (e^\eps - 1)} H_\oab (2\vecpsetemp-\mathbf{1}_\oab).
	\] \label{step:before_projection} 
	
	Keep the first $\ab$ elements of $\tp_\oab$, i.e., $\tp:=\tp_\oab^{[\ab]}$ and project it onto $\setsp$.
	\[
		\hp := \min_{\p \in \setsp} \norm{\tp - \p}{2}^2.
	\] \label{step:projection}
\end{algorithm}

\subsection{Upper bounds under LDP constraints} \label{sec:ldp_upper_a}

We now establish the sample and time complexity of Algorithm~\ref{alg:hadamard}. 

Steps~\ref{step:one}--\ref{step:before_projection}   of Algorithm~\ref{alg:hadamard} are identical to~\cite{AcharyaS19}, who showed a time complexity of $\tilde{O}(\ns+\ab)$. For the final step, where we project the vector $\tilde{p}_K$ on to $\setsp$, we can use~\cite[Algorithm 1]{kyrillidis2013sparse}, which runs in $\tilde{O}(\ab)$ time, proving the overall time complexity.

Algorithm~\ref{alg:hadamard} uses Hadamard matrices. For $m$ that is a 
power of two, let $H_m$ be the $m\times m$ Hadamard matrix with entries 
in $\{-1,1\}$. The privacy guarantee of the algorithm follows 
from~\eqref{equ:scheme}, which obeys~\eqref{eqn:lsp}. 
A key property we use is the following claim from~\cite{AcharyaS19}, which shows a relationship between underlying distribution and the message distributions. 

\begin{claim}[\cite{AcharyaS19}]	\label{clm:hr}
In~\eqref{equ:scheme}, let $\pset_j := \condprob{Y_i = 1 }{i \in S_j}$	for $j \in [\oab]$. Let $\vecpset:=(\pset_1, \ldots, \pset_{\oab})$. Let $\p_K$ be the distribution over $[K]$ obtained by appending $K-\ab$ zeros to $\p$. Then, 
		\begin{align}
			\p_\oab = \frac{(e^\eps+1)}{\oab (e^\eps - 1)} H_\oab (2\vecpset - \mathbf{1}_\oab).\label{eqn:hr-performance}
		\end{align}
\end{claim}
By definiton of $\hp$, we have $\norm{\tp - \hp}{2}^2 \le \norm{\tp - \p}{2}^2$. Hence
\[
\norm{\tp - \p}{2}^2 \ge \norm{\tp - \hp}{2}^2 = \norm{\tp - \p}{2}^2 + \norm{\p - \hp}{2}^2  + 2 \langle\tp - \p, \p - \hp\rangle.
\]
Rearranging the terms, we have
\begin{equation} \label{eqn:key_inequality}
\norm{\hp-\p }{2}^2 \le 2 \langle\tp - \p, \hp - \p\rangle.
\end{equation} 
We bound the right hand side by analyzing the projection step (Step~\ref{step:projection}), and using Claim~\ref{clm:hr}. 
 The proof of the lemma is from standard covering number arguments from 
 high dimensional 
 sparse regression~\citep{raskutti2011minimax}, and is provided in 
 Appendix~\ref{app:sparse-ub}.  
\begin{lemma}\label{lem:sparse}
	\[
	\langle \tp - \p, \hp - \p\rangle \le  { \frac{25(e^\eps+1)}{(e^\eps-1)}\frac{\sqrt{s \log(2\ab/s)} }{\sqrt{\ns}} \norm{\hp - p}{2}}.
	\]
\end{lemma}
We can now prove the sample complexity bound as follows. 
\begin{align}
	\dtv{\hp}{\p}=\frac{1}{2}\norm{\hp - \p}{1} 	&\le \frac{1}{2}\sqrt{2s}\norm{\hp - \p}{2}\label{eqn:csi} \\
	&\le \frac{40s \sqrt{\log(2\ab/s)}}{\sqrt{\ns}}\frac{e^\eps+1}{e^\eps-1},\label{eqn:sc-bound}
\end{align}
where~\eqref{eqn:csi} applies Cauchy-Schwarz inequality on the $2s$-sparse vector $\hp-p$, and~\eqref{eqn:sc-bound} is from plugging Lemma~\ref{lem:sparse} in~\eqref{eqn:key_inequality}. Plugging in $\dtv{\hp}{p}=\alpha$, and using $e^\eps-1=O(\eps)$ for $\eps=O(1)$ gives us the desired sample complexity bound of $n=O\Paren{s^2 \max \{ \log (k/s), 1\}/{\alpha^2 \eps^2}}$.

%% file: lower_LDP.tex
\subsection{Lower bound under LDP constraints} \label{sec:ldp_lower}
We now prove the sample complexity lower bound in Theorem \ref{thm:ldp_sparse} using the chi-squared contraction method in~\cite{AcharyaCT19} and an extension of Fano's method from~\cite{duchi2013distance}. 

For simplicity of analysis we add $0$ to the underlying domain and 
consider distributions over $[k]\cup\{0\}$. Let $\cZ_{\ab,s}\subseteq\{0, 
1\}^k$ be all $k$-ary binary strings with $s$ one's. Then, $|\cZ_{\ab,s}|={k 
\choose s}$.  We will restrict to $\cP_{\ab,s}:=\{p_z: z\in \cZ_{\ab, s}\}$, 
and $\p_z$ is described below for $z\in\cZ_{\ab,s}$:
\begin{equation}\label{eqn:pz}
	p_z(x)=\begin{cases}1-8\alpha,&\text{for }x=0,\\
	\frac{8\alpha z_x}{s}, &\text{for $x=1,\ldots, \ab$},
	\end{cases}
\end{equation}
where $z_x$ is the $x$th coordinate of $z$. Since $s$ of the $z_x$'s are 
one, $\sum_{x=1}^\ab p_z(x) = 8 \alpha$ and $\p_z$ is a valid distribution.

Let $Z :=(Z_1, \ldots, Z_k)$ be a uniform random variable over 
$\cZ_{\ab,s}$. Let $Y^\ns:=(Y_1, \ldots, Y_\ns)$ be the output of an 
$\eps$-LDP scheme whose input are $X^\ns= (X_1, \ldots, X_{\ns})$, 
drawn i.i.d. from $p_Z$, and $\hp$ is such that $\probof{\dtv{\p} 
{\hp(Y^\ns)} \le \alpha} \ge  0.9$. In other words, we can estimate 
distributions in $\cP_{\ab,s}$ to within $\alpha$ in total variation distance 
with probability at least 0.9. Let $\hat Z\in\cZ_{\ab,s}$ be such that 
$\p_{\hat Z}$ is the distribution in  $\cP_{\ab,s}$ closest to $\hp(Y^\ns)$ 
in $\tvd$. Then, we have 
\[
	\frac{4\alpha}{s} \ham{Z}{\hat{Z}}  = \dtv{\p_Z}{\p_{\hat{Z}}} \le \dtv{\hp}{\p_{\hat{Z}}} + \dtv{p_Z}{\hp} \le 2  \dtv{\p_Z}{\hp}.
\]
Since $\probof{\dtv{\p_Z}{\hp(Y^\ns)} \le \alpha} \ge  0.9$, we have $\probof{\ham{Z}{\hat{Z}}\le s/2} \ge 0.9$, which implies using the estimator $\hat p$, we can estimate the underlying $Z$ to within Hamming distance $s/2$. We now state a form of Fano's inequality from~\cite{duchi2013distance}, adapted to our setting. 

\begin{lemma}[Corollary 1 \cite{duchi2013distance}]
Let $\cZ_{\ab, s}\subseteq\{0,1\}^\ab$ and $Z$ be uniformly distributed over $\cZ_{\ab, s}$.  For  $t\ge 0$, define the maximum neighborhood size at radius $t$
\[
	N_t^{max}:=\max_{z\in\cZ}\{|z'\in\cZ:\ham{z}{z'}|\le t\},
\]
to be the maximum number of elements of $\cZ_{\ab, s}$ in a Hamming ball of radius $t$.
If $|\cZ_{\ab, s}|\ge 2 N_t^{max}$, then for any Markov chain $Z- Y^\ns- \hat{Z}$,
\begin{equation*}
\probof{\ham{\hat{Z}}{Z}>t}\ge 1-\frac{I(Z;Y^\ns)+\log 2}{\log|\cZ_{\ab, s}|-\log N_t^{max}}.
\end{equation*}
\label{lem:fano}
\end{lemma}

Substituting $t=s/2$, and using $\probof{\ham{Z}{\hat{Z}}\le s/2} \ge 0.9$ with this lemma gives 
\begin{align}
\frac{I(Z;Y^\ns)+\log 2}{\log|\cZ_{\ab,s}|-\log N_{s/2}^{max} }>0.9.
\label{eqn:bound-MI} 
\end{align}

Using chi-squared contraction bounds from~\cite{AcharyaCT19}, we upper bound $I(Z; Y^\ns)$ in the next lemma. 
\begin{lemma}
\label{lem:info}
Let $Z$ be uniformly drawn from $\cZ_{\ab, s}$ and $Y^\ns$ be the outputs of $n$ users, 
\begin{equation*}
I(Z;Y^\ns)=O\left(\frac{\ns\alpha^2(e^\eps-1)^2}{s}\right).
\end{equation*}
\end{lemma}
Next. we lower bound $\log {|\cZ_{\ab,s}|}-\log {N_{s/2}^{max}}$ using 
standard Gilbert-Varshamov type arguments in the following lemma.
\begin{lemma} Let $1\le s\le \ab/100$, then 
\begin{equation*}
\log {|\cZ_{\ab,s}|}-\log {N_{s/2}^{max}}\ge \frac s8\log \Paren{\frac {\ab}{s}}.
\end{equation*}
\label{lem:packing}
\end{lemma}

Plugging these two bounds in~\eqref{eqn:bound-MI} gives the tight sample complexity lower bound for sparse distribution estimation under $\eps$-local differential privacy. 

We now prove these lemmas.
\begin{proof}[Proof of Lemma~\ref{lem:info}]
For a distribution $q$ over $[\ab]\cup\{0\}$ and a channel $W$, let $q^W$ 
denote output distribution of $Y$ when $X \sim q$ and $Y=W(X)$. Then, 
\begin{align}
	q^W(y)=\sum_{x\in[\ab]\cup\{0\}}W(y \mid x)q(x).\label{eqn:out-message}
\end{align}

Let $p_0$ be the average of all distributions in $\cP_{\ab,s}$. Then  
$p_0(0)=1-8\alpha$, and $p_0(x)=8\alpha/k$ for $x=1,\ldots, \ab$. We 
will use chi square contraction bound in~\cite{AcharyaCT19} to bound the 
maximum value of $I(Z;Y^\ns)$ in terms of the $\chi^2$-divergence 
between the output distributions induced by distributions in $\cP_{\ab,s}$ 
and by $p_0$ and as follows:
\begin{align}
I(Z;Y^\ns)& \le  \ns\cdot \max_{W\in \cW_\eps}\expectover{Z}{\chi^2\left(\p_Z^W, \p_0^W\right)}\label{eqn:chi-sq}\\
	&= \ns \cdot\max_{W\in \cW_\eps}\sum_y \frac{\expectover{Z}{\Paren{\sum_{x=1}^\ab (\p_Z(x)-\p_0(x)) W(y \mid x)}^2}}{\expectover{X \sim \p_0}{W(y\mid X)}},
	\label{eqn:chi-square}
\end{align}
where~\eqref{eqn:chi-sq} is from the chi-squared contraction bound, 
and~\eqref{eqn:chi-square} is by using~\eqref{eqn:out-message} in the 
definition of $\chi^2$-divergence\footnote{$\chi^2 (p,q):=\sum_x 
(p(x)-q(x))^2/q(x)$.}.

For an $\eps$-LDP channel $W\in\cW_\eps$, let $W_{\rm min}^y:=\min_x W(y\mid x)$. By~\eqref{eqn:lsp}, we have $W(y\mid x) = W_{\rm min}^y +\eta^y_x\cdot W_{\rm min}^y$ for some $0 \le \eta^y_x \le e^{\eps}-1$. Furthermore, for $z \in\cZ_{\ab,s}$, by the definition of $p_z$,  $\p_z(x) -\p_0(x) = 8\alpha \Paren{{\frac{z_x}{s}}-\frac{1}{\ab}}$, and $\sum_x z_x=s$, thus giving
\begin{align*}
\sum_x (\p_z(x)\!-\!\p_0(x)) W(y\mid x) = 
8\alpha \sum_x \Paren{{\frac{z_x}{s}}\!-\!\frac{1}{\ab}}\Paren{W_{\rm 
min}^y +W_{\rm min}^y\eta^y_x}= 8\alpha W_{\rm min}^y \cdot\sum_x 
\Paren{{\frac{z_x}{s}} \!-\! \frac{1}{\ab}}\eta^y_x.
\end{align*}

Since $Z$ is uniformly distributed over $\cZ_{\ab,s}$, elementary computations show that  $\expectation{Z_x}=\expectation{Z_x^2}=s/\ab$, and for $x_1\ne x_2 \in[\ab]$,
$	\expectation{Z_{x_1}Z_{x_2}} = {\ab-2 \choose s-2}\big/{\ab\choose s}=\frac{s(s-1)}{\ab(\ab-1)}.$

Therefore, 
\begin{align*}
&\ \expectover{Z}{\Paren{\sum_x (\p_Z(x)-\p_0(x)) W(y\mid x) }^2}\\
&= 64\alpha^2 (W_{\rm min}^y )^2 \cdot \Paren{\sum_{x_1,x_2}\expectover{Z}{\frac 1{\ab^2} -\frac{Z_{x_1}+Z_{x_2}}{s\ab}+\frac{Z_{x_1}Z_{x_2}}{s^2}}\eta^y_{x_1}\eta^y_{x_2}}\\
&= 64\alpha^2 (W_{\rm min}^y )^2 \Paren{\sum_{x}\left[\frac{1}{s\ab}-\frac 1{\ab^2}\right](\eta^y_x)^2+\sum_{x_1\ne x_2}\expectover{Z}{-\frac 1{\ab^2}+\frac{s-1}{s\ab(\ab-1)}}\eta^y_{x_1}\eta^y_{x_2}}\\
&\le 64\alpha^2 (W_{\rm min}^y )^2 \Paren{\frac{(\max_x \eta^y_x)^2}{s}},
\end{align*}
and 
\begin{align*}
\frac{\expectover{Z}{\Paren{\sum_x (\p_Z(x)-\p_0(x)) W(y\mid x) }^2}}{\expectover{X \sim \p_0}{W(y|X)}} \le 64\alpha^2  \Paren{\frac{(\max_x \eta^y_x)^2}{s}}\cdot W_{\rm min}^y .
\end{align*}

Using $\sum_y W_{\rm min}^y \le 1$, and $\eta_x^y \le e^\eps-1$, we obtain 
\begin{align*}
\expectover{Z}{\chi^2(\p_Z^W,\p_0^W)} = O\Paren{\frac{\alpha^2 (e^\eps-1)^2}s}.
\end{align*}
Combining with \eqref{eqn:chi-square}, this completes the proof.
\end{proof}

\begin{proof}[Proof of Lemma~\ref{lem:packing}]
Let $z \in \cZ_{\ab,s}$. A vector in $\cZ_{\ab,s}$ that is at most $s/2$ 
away from $z$ in Hamming distance can be obtained as follows: Fix $s/2$ 
coordinates in $z$ that are 1, and from the remaining $\ab-s/2$ 
coordinates, choose $s/2$ coordinates and make them 1. All other 
coordinates are set to zero. Therefore,
\[
	N_{s/2}^{max}\le {s \choose s/2}{\ab-s/2\choose s/2}.
\]
Recall that $|\cZ_{\ab,s}|={\ab\choose s}$. Using Stirling's approximation for binomial coefficients\footnote{For $ 1\le s\le k$, we have $\left(\frac{k}{s}\right)^s\le {k\choose s}\le \left(\frac{ke}{s}\right)^s$.},
we get
\begin{align*}
\log \frac{|\cZ|}{N_{s/2}^{max}}\! \ge \!\log \frac{{\ab \choose s}}{ {s 
\choose s/2}{\ab-s/2\choose s/2}} 
\!\ge\! \log 
\frac{\left(\frac{\ab}{s}\right)^s}{\left(2e\right)^{s/2}\left(\frac{(2\ab-s)e}{s}\right)^{s/2}}
\!\ge\! \log  
\frac{\left(\frac{\ab}{s}\right)^s}{\left(2e\right)^{s/2}\left(\frac{(2\ab)e}{s}\right)^{s/2}}
\!=\! \frac{s}{2}\log\left(\frac{\ab}{4e^2s}\right)\!,
\end{align*}
which is at least $\frac s8\log \frac{\ab}{s}$ when $s\le \ab/100$.
\end{proof}

%% file: upper_communication_new.tex
We now prove guarantees of Theorem~\ref{thm:lbit_sparse} and establish 
upper and lower bounds for communication constrained sparse discrete 
distribution estimation. In Section~\ref{sec:lbit_upper}, we propose an 
algorithm that requires public randomness with sample complexity given in 
Theorem~\ref{thm:lbit_sparse}. Designing a private-coin protocol for 
estimation is an open question. In Section~\ref{sec:lbit_lower} we establish 
the lower bounds. 

\subsection{Upper bounds under communication constraints} \label{sec:lbit_upper}
Note that the sample complexity upper bound in Theorem~\ref{thm:lbit_sparse} has $\min\left\{2^\ell, s\right\}$, which equals $s$ when $\ell\ge\log s$. 
We therefore only consider $\ell\le\log s$ since if $\ell >\log s$, we can just use $\log s$ bits and get the same bound.

Our first step is to use public randomness to design hash functions at the users. 	A randomized mapping $h: [\ab] \rightarrow [2^{\ell}]$ is a \emph{random hash function} if $\forall x \in[\ab], y \in [2^\ell]$,
	\[
		  \probof{h(x) = y} = \frac{1}{2^{\ell}}.
	\]

\noindent\textbf{The scheme.} Let $h_1, \ldots, h_\ns$ be $\ns$ independent hash functions, available at the users and at the server. User $i$'s $\ell$ bit output is $Y_i=h_i(X_i)\in[2^\ell]$.
The probability of $x\in[\ab]$ being in the preimage of user $i$'s message 
$Y_i$ is, 
\begin{align}\label{eqn:pre-image}
 \probof{Y_i = h_i(x)}= \p(x)+(1-\p(x))\frac1{2^\ell} = 
 \p(x)\left(1-\frac{1}{2^\ell}\right)+\frac{1}{2^\ell}=: \bias(x).
 \end{align}

\noindent\textbf{The estimator.} Upon receiving messages $Y_1,\ldots, 
Y_\ns$, the estimator is as follows, 

\begin{enumerate}
	\item
	The first $\nsein$ messages are used to obtain a set $T\subseteq[k]$ 
	with $|T|=O(s)$ such that with high probability $p(T)>1-\alpha/2$.  
	\item 
	 With the remaining messages we estimate $p(x)$ for $x\in T$. 
\end{enumerate}

\noindent We now describe and analyze the two steps. 

\noindent\textbf{Step 1.} Let $Y_1, \ldots, Y_{\nsein}$ be the first $n/2$ 
messages, where $Y_i = h_i(X_i)$. For $x=1,\ldots, \ab$, let 
\[
\ms(x):=|\{i:h_i(x)=Y_i, 1\le i \le \nsein\}|
\]
 be the number of these messages in the first half whose preimage $x$ 
 belongs to. Let $T\subseteq [\ab]$ be the set of symbols with the largest 
 $|T|= 2s$ values of $\ms(x)$'s. If $\p(x)$ is large we expect $\ms(x)$ to 
 be large. In particular, we show that for sufficiently large $n$, the 
 probability of symbols not in $T$ is small.
\begin{lemma}
\label{lem:step-one}
There is a constant $C_1$ such that for $n = C_1\cdot{s^2\log 
(\ab/s)}/({\alpha^2 \min\{2^\ell,s\}})$ with probability at least $0.95$, 
$p(T):=\sum_{x\in T}p(x) \ge 1-\alpha/2$.
\end{lemma}
\noindent\textbf{Step 2.}  For $x=1,\ldots, \ab$, let $N(x):=|\{i:h_i(x)=Y_i, 
\nsein<i\le \ns \}|$ be the number of messages  in the second half such 
that $x$ belongs to the preimage of $Y_i$. Our final estimator is given by
\begin{equation}\label{eqn:phat-comm}
\hp(x) = \begin{cases}
	\frac{(2^\ell N(x)/(\nsein))-1}{2^\ell-1},& \text{if $x\in T$}\\
	0, &\text{otherwise.}
\end{cases}
\end{equation}

The following lemma shows that $\hp$ converges to $p$ over $T$. 
\begin{lemma}
	\label{thm:utility_com}
	There is a constant $C_2$ such that for $n= C_2\cdot s^2/({\alpha^2 
	\min\{2^\ell,s\}})$ with probability at least $0.95$,
	\[
		\sum_{x\in T} |\hp(x)-p(x)| \le \frac{\alpha}{2}.
	\]
\end{lemma}

Combining Lemma~\ref{lem:step-one} and Lemma~\ref{thm:utility_com}, by the union bound we get that with $n$ samples, with probability at least $0.9$, $\norm{\hp-p}{1}\le\alpha$, establishing upper bound of Theorem~\ref{thm:lbit_sparse}.

%% file: lower_communication.tex
\subsection{Lower bounds under communication constraints} \label{sec:lbit_lower}
We now prove the sample complexity lower bound in Theorem~\ref{thm:lbit_sparse}. The first term of $\Omega(\frac{s^2}{\alpha^2 \min\{2^\ell, s\}})$ follows from~\eqref{eqn:sample-complexity} and holds even with the knowledge of the support $S$. 

We prove the second term by considering the construction $\cP_{\ab,s}$ as in Section~\ref{sec:ldp_lower}, and bound~\eqref{eqn:bound-MI}. The lower bound on $\log(|\cZ|/N_t^{max})$ is the same as from Lemma~\ref{lem:packing}. Analogous to Lemma~\ref{lem:info}, we will now bound the mutual information $I(Z;Y^\ns)$ by $O(\ns\alpha 2^\ell/s)$ as follows. As in~\eqref{eqn:chi-square}, we have the following bound
\begin{align}
I(Z;Y^\ns)&\le  \ns \cdot\max_{W\in \cW_\ell}\sum_y \frac{\expectover{Z}{\Paren{\sum_{x=1}^\ab (\p_Z(x)-\p_0(x)) W(y \mid x)}^2}}{\expectover{X \sim \p_0}{W(y\mid X)}},
	\label{eqn:chi-square-comm}
\end{align} 
and we now bound it for each $y\in[2^\ell]$. Similar to the expansion in proving the LDP lower bounds, we have
\begin{align*}
	&\quad\expectover{Z}{\Paren{\sum_x (\p_Z(x)-\p_0(x)) W(y\mid x)}^2}\\
	&= 64\alpha^2\expectover{Z}{\sum_{x_1, x_2}\Paren{\frac{1}{k^2}-\frac{Z_{x_1}+Z_{x_2}}{sk}+\frac{Z_{x_1}Z_{x_2}}{s^2}}W(y\mid x_1)W(y\mid x_2)}\\
	&= 64\alpha^2\sum_{x=1}^k\left(\frac{1}{sk}-\frac{1}{k^2}\right)W(y\mid x)^2+64\alpha^2 \sum_{x_1\ne x_2}\left(-\frac{1}{k^2}+\frac{s-1}{sk(k-1)}\right) W(y\mid x_1)W(y\mid x_2) \\
	&	\le  64\alpha^2\sum_{x=1}^k\left(\frac{1}{sk}-\frac{1}{k^2}\right)W(y\mid x)^2.
	\end{align*}
Note that 
$\expectover{X \sim p_0}{W(y\mid X)}=(1-8\alpha)W(y\mid 0)+\sum_{x=1}^k\frac{8\alpha}{k}W(y\mid x).$
Hence,
\begin{align*}
    \expectover{Z}{\chi^2(p_Z^W, p_0^W)}\le & 64\alpha^2\sum_{y}\frac{\Paren{\sum_{x=1}^k\left(\frac{1}{sk}-\frac{1}{k^2}\right)W(y\mid x)^2}}{(1-8\alpha)W(y\mid 0)+\sum_{x=1}^k\frac{8\alpha}{k}W(y\mid x)}\\
    &\le 64\alpha^2 \left(\frac{1}{sk}-\frac{1}{k^2}\right)\sum_{y}\frac{\Paren{\sum_{x=1}^kW(y\mid x)^2}}{\sum_{x=1}^k\frac{8\alpha}{k}W(y\mid x)}\\
    &\le \frac{8\alpha}{s}\sum_y \frac{\sum_{x=1}^kW(y\mid x)^2}{\sum_{x=1}^kW(y\mid x)}\\
    &\le \frac{8\alpha}{s}2^\ell,
\end{align*}
where we used that $W(y|x)^2\le W(y|x)$, proving the lower bound.

%% file: experiments.tex
To verify our bounds, we evaluate our algorithms on sythetic datasets. We 
fix the support size $k=5000$ and draw the data from
distributions uniform over a subset with size $s$ which takes values in 
$2^i, 
i=1, \ldots, 12$. 

For LDP, we set $\eps=0.5, 0.9$ and draw $n=3\times10^6$ samples. 
For communication, we set $\ell=1, \ldots, 7$ and $n=400000$. The 
estimation errors in $d_{TV}$ with respect to sparsity $s$ are shown in Figure 
\ref{fig:all_exp}. In both experiments, we observe significant increase in 
accuracy when the sparsity decreases. \newzs{And in both experiments, we 
observe that under the same sparsity, larger $\eps$ (LDP) or $\ell$ 
(communication) leads to a better 
utility, which is consistent with our theoretical analysis.}

Moreover, for LDP we compare the Hadamard Response algorithm with 
sparse projection \newzs{(exact procedures in 
Algorithm~\ref{alg:hadamard})} 
and 
regular projection \newzs{(projecting onto $\setk$ instead of $\setsp$ in 
Step~\ref{step:projection} of Algorithm~\ref{alg:hadamard})}. \newzs{The 
results 
show that sparse 
projection, which needs the knowledge of $s$, leads to a significant 
improvement in utility. Straightforward modifications of our proof shows 
that projecting onto $\setk$, which 
doesn't need to know $s$ in advance, will lead to a risk bound of 
$O(s\sqrt{\log k}/\sqrt{n})$ instead of $O(s\sqrt{\log (k/s)}/\sqrt{n})$ in 
Theorem~\ref{thm:ldp_sparse}. Whether this gap is inevitable is an 
interesting direction to explore.}
\begin{figure}
    \centering
    \includegraphics[width=0.45\linewidth]{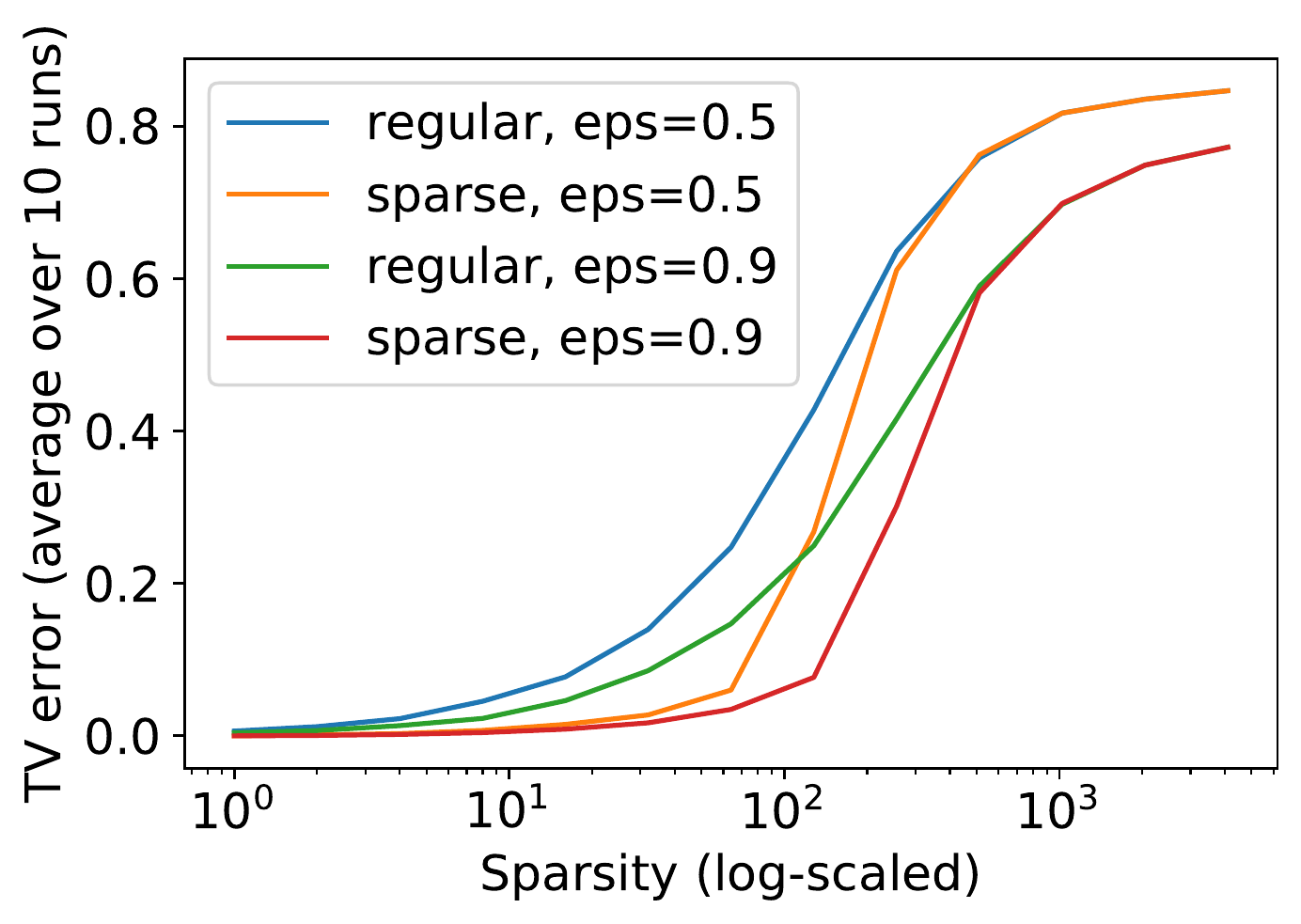}\includegraphics[width=0.45\linewidth]{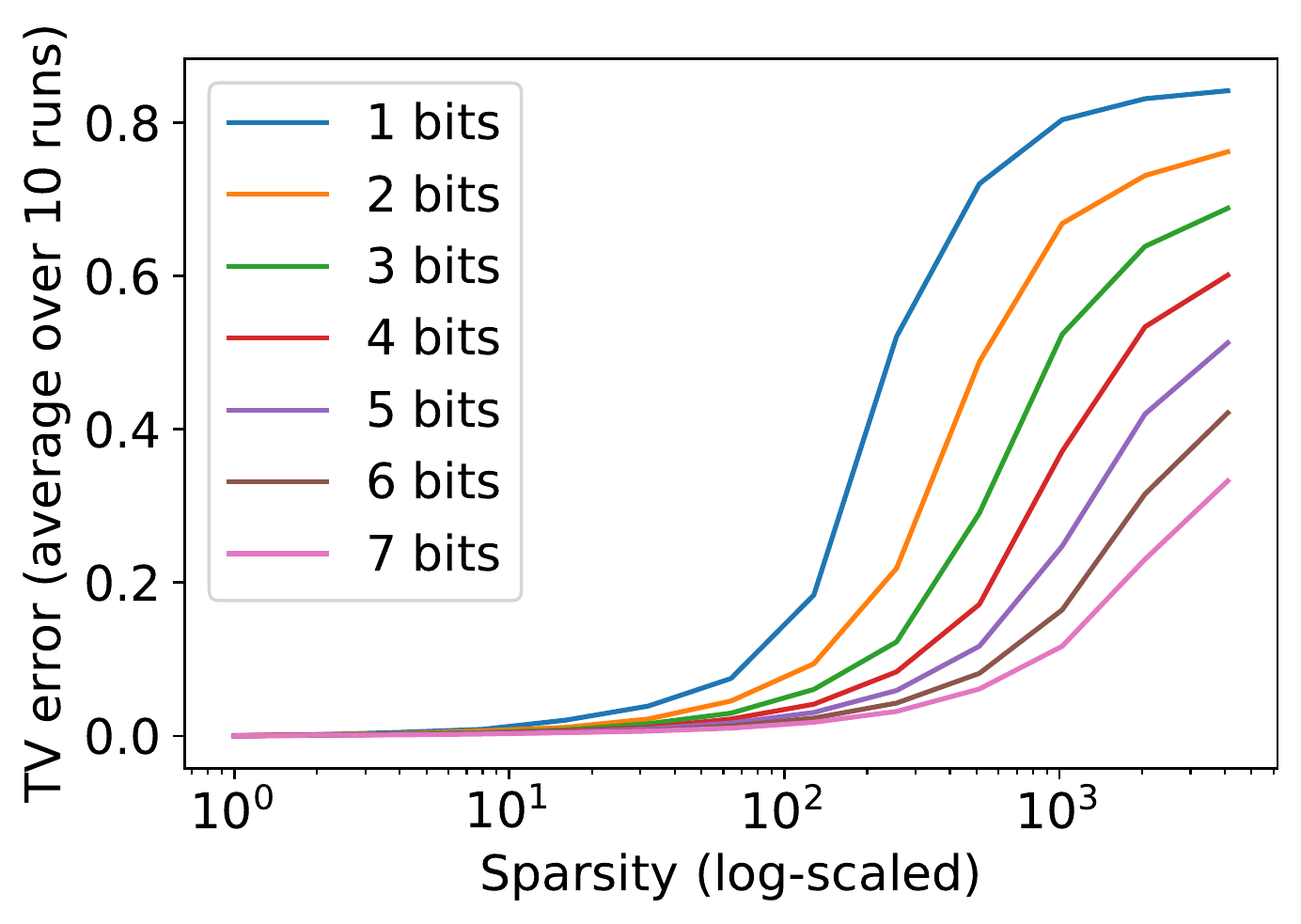}
    \caption{Estimation errors of the proposed algorithms under different 
    support sizes. \textbf{Left}: LDP (comparing regular and sparse projections); \textbf{Right}: communication 
    constraints.}
    \label{fig:all_exp}
\end{figure}

%% file: app-lemma-covering.tex
We will use the following bound on the covering number of $s$-sparse vectors. 
\begin{claim}[\cite{raskutti2011minimax}]
	Let $S(r, 2s) = \{ \err \in \RR^k:  \norm{\err}{2} \le r,  \norm{\err}{0} \le 2s\}$. There exists a $\rho r$-covering (in $\ell_2$) $\cC_r \subset S(r, 2s)$ of $S(r, 2s)$  with size
	\[
	|\cC_r|=N(S(r, 2s), \rho r) \le {\ab \choose 2s} \Paren{\frac{1}{\rho}}^{2s}.
	\]
\end{claim}
The bound follows from the fact that ${k\choose 2s}$ $2s$-dimensional 
subspaces 
are sufficient to cover $S(r, 2s)$, and each subspace can be covered with 
$(1/\rho)^{2s}$ $\ell_2$ balls of radius $\rho r$.

	Let $\cC_1$ be a $\rho$-covering of $S(1, 2s)$ with size $N(S(1, 2r), 
	\rho)$. Let $r = \norm{\hp - \p}{2}$, then $\hp - \p \in S(r, 2s)$, and 
	one can obtain a covering $\cC_r$ of $S(r, 2s)$ by multiplying each 
	vector in $\cC_1$ by $r$.
	Since $\hp - \p$ is $2s$-sparse, let $\err^*$ be the closest point in $\cC_r$ to $\hp - p$. Then we can bound the right hand side of~\eqref{eqn:key_inequality}.
	\begin{align*}
		\langle \tp - \p, \hp - \p\rangle & \le |\langle\tp - \p, \err^* \rangle| + |\langle\tp - \p, \hp - \p - \err^* \rangle| \\
		& \le \max_{\err \in \cC_r}|\langle \tp - \p, \err \rangle| + \rho r \norm{\tp - \p}{2} \\
		& = r\max_{\err \in \cC_r}|\langle \tp - \p, \err/r \rangle| + \rho r \norm{\tp - \p}{2} \\
		& = r\max_{\err \in \cC_1}|\langle \tp - \p, \err \rangle| + \rho r \norm{\tp - \p}{2}.
	\end{align*}
	$\tp - \p$ is the first $\ab$ entries of $\tp_\oab - \p_\oab$, and by Claim~\ref{clm:hr}, we have
	\[
		\tp_\oab - \p_\oab = \frac{2(e^\eps+1)}{\oab (e^\eps - 1)} H_\oab (\vecpsetemp-\vecpset).
	\]
	Note that in Algorithm~\ref{alg:hadamard}, 	$\forall j \in [\oab]$, 
	$\psetemp_j$ is the average of $|S_j| \ge \oab/(2n)$ i.i.d. Bernoulli 
	random variables with the mean satisfying $\expectation{\vecpsetemp} 
	= \vecpset$. Hence $\forall j \in [\oab], \psetemp_j- \pset_j$ is a 
	zero-mean sub-Gaussian random variable with variance proxy at most 
	$\frac{\oab}{2\ns}$. Moreover, $\psetemp_j$'s are independent since 
	$S_j$'s are disjoint. %
	Note for all $\forall \err\in \cC_1$, 
	\[
		\langle \tp_\oab - \p_\oab, \err\rangle = \frac{2(e^\eps+1)}{\oab (e^\eps - 1)}  \err^T H_\oab (\vecpsetemp-\vecpset),
	\]
	which are linear combinations of $(\psetemp_j - \pset_j)$'s. Since $\tp - \p$ is the first $k$ entries of $\tp_\oab - \p_\oab$, we have $\forall \err \in \cC_1$, $\langle \tp - \p, \err \rangle$ is also sub-Gaussian (see Corollary 1.7 in~\cite{rigollet2015high}) with variance proxy at most
	\begin{align}
	\frac{\oab}{2\ns} \norm{\frac{2(e^\eps+1)}{\oab (e^\eps - 1)}  \err^T H_\oab}{2}^2 &= \frac{2}{\ns \oab} \Paren{\frac{e^\eps+1}{e^\eps-1}}^2    \norm{\err^T H_\oab}{2}^2 \nonumber\\
	&= \frac{2}{\ns} \Paren{\frac{e^\eps+1}{e^\eps-1}}^2    \norm{\err}{2}^2 \label{eqn:orthoganility} \\
	&\le \frac{2}{\ns}\left(\frac{e^\eps+1}{e^\eps-1}\right)^2 =:\sigma^2,\nonumber
	\end{align}
	where~\eqref{eqn:orthoganility} follows since $H_m^TH_m=m\II_k$ by the orthogonality of Hadamard matrices.

	Therefore using maximal inequalties of sub-Gaussian random variables~\cite[Theorem 1.14]{rigollet2015high}, with probability at least $19/20$, we have 
	\begin{equation} \label{eqn:max_subg}
		\max_{e \in \cC_1}|\langle \tp - \p, \err \rangle| < \sigma \sqrt{14\log|\cC_1|}.
	\end{equation}
	By the utility guarantee of Hadamard Response~\citep{AcharyaS19}, with 
	probability at least $19/20$, 
	\begin{equation}\label{eqn:l2_bound}
		\norm{\tp- \p}{2}^2\le 20\expectation{\norm{\tp- \p}{2}^2}\le \frac{40k(e^\eps+1)^2}{\ns(e^\eps-1)^2}\implies \norm{\tp- \p}{2} \le 7\frac{e^\eps +1 }{e^\eps -1 }\sqrt{\frac{\ab}{\ns}}.
	\end{equation}
    By union bound, conditioned on~\eqref{eqn:max_subg} and~\eqref{eqn:l2_bound}, which happens with probability at least $9/10$,
	\begin{align*}
		\langle \tp - \p, \hp - \p\rangle &\le \frac{e^\eps+1}{e^\eps-1}\frac{r}{\sqrt{\ns}} \sqrt{56\log |\cC_1|} + 7\rho r \frac{e^\eps+1}{e^\eps-1}\sqrt{\frac{\ab}{\ns}}  \\
		& \le \frac{e^\eps+1}{e^\eps-1}\frac{r}{\sqrt{\ns}} \sqrt{112s \log\left(\frac{e\ab}{2s}\right) + 112s \log \frac{1}{\rho}}  +7\rho r \frac{e^\eps+1}{e^\eps-1}\sqrt{\frac{\ab}{\ns}}.
	\end{align*}
	Taking $\rho = \sqrt{\frac{s}{\ab}}$, 
	\[
	\langle \tp - \p, \hp - \p\rangle \le 25\frac{e^\eps+1}{e^\eps-1}\frac{r 
	\sqrt{s\log(2\ab/s)} }{\sqrt{\ns}} = 
	25\frac{e^\eps+1}{e^\eps-1}\frac{\sqrt{s \log(2\ab/s)} }{\sqrt{\ns}} 
	\norm{\hp - p}{2},
	\]
	concluding the proof.

%% file: app-comm-lemmas.tex
\begin{proof}[Proof of Lemma~\ref{lem:step-one}]
$\ms(x)$ is distributed as a Binomial ${Bin}(\ns/2,\bias(x))$. Hence, $\expectation{\ms(x)}=\ns \cdot \bias(x)/2$. By~\eqref{eqn:pre-image} and~\eqref{eqn:phat-comm} we have $\expectation{\hp(x)}=p(x)$. Using the variance formula of Binomials, we know $\variance{\ms(x)}=\ns\cdot \bias(x)\cdot(1-\bias(x))/2\le \ns\cdot \bias(x)/2$.

Set $\gamma=1-1/2^\ell$ and $\beta=1/2^\ell$, then $\bias(x)= \gamma\cdot p(x)+\beta$. We will use the following multiplicative Chernoff bound. 
\begin{lemma}[Multiplicative Chernoff 
bound~\citep{mitzenmacher2017probability}] \label{lem:chernoff}
Let $Y_1, \ldots, Y_n$ be independent random variables with $Y_i\in{0,1}$, and $Y = Y_1+\dots+Y_\ns$, and $\mu= \expectation{Y}$. Then for $\tau>0$, 
\begin{align*}
    \probof{Y\ge(1+\tau)\mu}\le e^{-\frac{\tau^2\mu}{2+\tau}},\quad    \probof{Y\ge(1-\tau)\mu}\le e^{-\frac{\tau^2\mu}{2}}.
\end{align*}
\end{lemma}

Let $S:=\{x:p(x)>0\}$. Therefore, for $x \in [\ab]\setminus S$, $p(x)=0$. By Lemma~\ref{lem:chernoff},
\[
\probof{\ms(x)\ge \nse\beta +\sqrt{3 \ns \beta\log \frac \ab s}}\le \Paren{\frac s\ab}^{2}.
\]

Let $E$ be the event that at most $s$ symbols in $[\ab]\setminus S$ appear at least $\ms^\ast :=\nse\beta + \sqrt{3\ns \beta \log \frac ks}$ times. By Markov's inequality,
\begin{align} \label{eqn:ec}
\probof{E^c}=\probof{\left|x \in [\ab]\setminus S: \ms(x)\ge \nse\beta + \sqrt{3\ns \beta \log \frac \ab s}\right|>s}\le \frac s\ab \le \frac{1}{100}.
\end{align}

We condition on $E$ in the remainder of the proof. Note that it suffices to show
\begin{equation} \label{eqn:sum_missing}
	\condexpect{\p(T^c)}{E} :=\sum_x\p(x)\condprob{x \text{ not selected}}{E}  \le \frac{\alpha}{50},
\end{equation}
since if~\eqref{eqn:sum_missing} is true, by Markov inequality, 
\[
\condprob{\p(T^c) > \frac{\alpha}{2}}{E} \le \frac{1}{25}.
\]
which, combined with~\eqref{eqn:ec}, implies Lemma~\ref{lem:step-one}.

Next we prove~\eqref{eqn:sum_missing}. Conditioned on $E$, a symbol $x$ is not selected after the first stage only if it appears at most  $\ms^\ast$ times, which implies $\condprob{x \text{ not selected}}{E}\le \condprob{\ms(x) \le \ms^\ast}{E}$. 
Moreover, since $\forall x_1 \in S, x_2 \notin S$, $\ms(x_1)$ and $\ms(x_2)$ are independent, we have $\ms(x_1)$ is independent of event $E$. Thus:

\begin{align*}
	\condexpect{\p(T^c)}{E} =\sum_x\p(x)\condprob{x \text{ not selected}}{E} \le 
	\sum_x\p(x)\probof{\ms(x) \le \ms^\ast}.
\end{align*}
Next we divide symbols into three sets based on their probability mass: $A = \{ x \in [k]: p(x) \le \frac{\alpha}{60s}\}$, $B = \{ x \in [k]:  \frac{\alpha}{60s} < p(x) \le \beta/\gamma \}$ and $C = \{ x \in [k]:  p(x) > \beta/\gamma \}$. For set $A$, we have:
\begin{align}\label{eqn:bound_a}
	\sum_{x \in A} \p(x)\probof{\ms(x) \le \ms^\ast} \le \sum_{x \in A} \p(x) \le \frac{\alpha}{60}.
\end{align}

Next we bound the sum over set $B$ and $C$. $\forall x \in B \cup C, \p(x)>\alpha/60s$. In the rest of the proof, we set the constant $C_1 = 700000$. For $n= C_1\cdot{s^2\log (\ab/s)}/({\alpha^2 \min\{2^\ell,s\}})$,
\begin{align*}
\frac{\ns}{4} \gamma\p(x)- \sqrt{3 \ns \beta \log \frac \ab s}\ge 0. 
\end{align*}
Hence,
\begin{align*}
\expectation{\ms(x)-\ms^\ast}=\nse \gamma\p(x)- \sqrt{3 \ns \beta \log \frac \ab s} \ge \frac{\ns\gamma\p(x)}{4},
\end{align*}

Using Lemma~\ref{lem:chernoff},
\begin{align*}
\probof{\ms(x) \le \ms^\ast} 
&= \probof{\ms(x) \le \expectation{\ms(x)} -(\expectation{\ms(x)} -\ms^\ast)}\\
&\le \exp\Paren{-\left(\frac{\gamma p(x)/2}{\gamma p(x)+ \beta}\right)^2 \nse(\gamma p(x) + \beta)}\\
& = \exp\Paren{-\frac{\gamma^2 p(x)^2\ns}{8(\gamma p(x) + \beta)} } 
\end{align*}
If $x \in C$, i.e., $p(x) > \beta/\gamma$, we have
\[
	\probof{\ms(x) \le \ms^\ast}  \le  \exp\Paren{-\frac{\gamma^2 p(x)^2\ns}{16\gamma p(x)} } = \exp\Paren{-\frac{\gamma p(x)\ns}{16} }  \le \frac{16}{\gamma p(x)\ns} \le \frac{16}{\ns \beta} \le \frac{\alpha}{1000},
\]
where we use $n\beta > C_1 s^2\log(\ab/s)/(\alpha^2 2^{2\ell})\ge C_1/\alpha^2$ when $s\ge 2^\ell$. This implies
\begin{align}\label{eqn:bound_b}
	\sum_{x \in C} \p(x)\probof{\ms(x)<\ms^\ast} \le \sum_{x \in A} \p(x) \frac{\alpha}{1000}\le \frac{\alpha}{1000}.
\end{align}

\noindent When $x \in B$, i.e, $p(x) \le  \beta/\gamma$,
\[
    \probof{\ms(x) \le \ms^\ast} \le \exp\Paren{-\frac{\gamma^2 p(x)^2\ns}{16\beta} }.
\]
Now let $\p(x)=(1+\zeta_x)\alpha/60s$ where $\zeta_x>0$, we have 
\begin{align*}
\frac{\gamma^2 p(x)^2\ns}{16\beta} \ge 2(1+\zeta_x)^2 \log \frac \ab s.
\end{align*}
\begin{align}\label{eqn:bound_c}
\sum_{x\in C}p(x)\probof{\ms(x) \le \ms^*} \le 
\frac \alpha{60s}\sum_{x\in C} (1+\zeta_x)\exp\Paren{-2 (1+\zeta_x)^2\log \frac \ab s}\le \frac {\alpha}{500}.
\end{align}
Combining~\eqref{eqn:bound_a}, \eqref{eqn:bound_b} and \eqref{eqn:bound_c}, we get~\eqref{eqn:sum_missing}, and thus proving the lemma.
\end{proof}

\begin{proof}[Proof of Lemma~\ref{thm:utility_com}]
Note that $N(x)$ and $\ms(x)$ are identically distributed, and therefore,  $\ms(x)$ is distributed ${Bin}(\ns/2,\bias(x))$, and $\expectation{N(x)}=\ns \cdot b(x)/2$, and $\variance{N(x)}\le \ns\cdot \bias(x)/2$. 
\begin{align}
\expectation{(\hp(x)-p(x))^2} = \left(\frac{2\cdot2^\ell}{\ns(2^\ell-1)}\right)^2 \cdot \variance{N(x)} \le \frac{2}{\ns}\left(\frac{2^\ell}{2^\ell-1}\right)^2 \bias(x).
\end{align}
Now, note that $\sum_{x\in T} \bias(x) = p(T)(1-1/2^\ell)+|T|/2^\ell$, and therefore,

\begin{align*}
\expectation{\norm{\hp^T-p^T}{2}^2} = \sum_{x\in T}\expectation{(\hp(x)-p(x))^2}%
\le \frac{2}{\ns}\left(\frac{2^\ell}{2^\ell-1}\right)^2\sum_{x\in T} \bias(x)
\le \frac{2(|T| + 2^\ell)2^\ell}{\ns (2^\ell-1)^2}.
\end{align*}
Using Jensen's inequality and Cauchy-Schwarz,
\[
	\expectation{\norm{\hp^T - \p^T}{1}}\le \sqrt{\expectation{\norm{\hp^T - \p^T}{1}^2} }\le \sqrt{|T|\cdot\expectation{\norm{\hp^T - \p^T}{2}^2}} \le \sqrt{\frac{4 s 2^\ell (2^\ell + 2s) }{\ns(2^\ell-1)^2}}.
\]
Setting $C_2 = 6400$, the lemma follows by Markov's inequality.
\end{proof}

%% file: appendix.tex
\section{An LDP estimation scheme using RAPPOR}
\smargin{Do we want to remove this section?}
\label{app:rappor}
We first describe the high level idea of the algorithm for LDP estimation. All 
users send their privatized data using RAPPOR~\citep{ErlingssonPK14, 
KairouzBR16}. As in Section~\ref{sec:lbit_upper}, we use the first half of 
privatized samples to estimate a subset $T\subseteq[\ab]$ with size $O(s)$ 
which contains most of the probability densities; we then use the remaining 
samples to estimate the distribution only on this set $T$. Details are 
described in Algorithm~\ref{alg:rappor}.

\begin{algorithm}[h]
\DontPrintSemicolon
\SetAlgoLined
\LinesNumbered
\KwIn{$\ns$ i.i.d. samples from unknown $s$-sparse $\p$.}
Each user randomizes its sample using RAPPOR: Each sample $X_i$ is first 
mapped to a one-hot vector $Z_i\in \{0, 1\}^\ab$ which has a 1 at the 
$X_i$'th coordinate and 0's elsewhere. Then each bit is flipped 
independently with probability $1/(e^{\eps/2}+1)$ to obtain $Y_i\in\{0, 
1\}^\ab$

Compute $\ms:=[\ms(1), \ldots, \ms(\ab)]=\sum_{i=1}^{\nsein}Y_i$ using 
the first $\nse$ samples.

Construct the set $T\subseteq[n]$ by keeping the $2s$ symbols with 
highest $\ms(x)$'s.

Obtain $\hp$: estimate the distribution over $T$ using the remaining $\nse$ samples. 
\caption{Sparse estimation using RAPPOR}
\label{alg:rappor}
\end{algorithm}

We note that
\[
\expectation{\ms(x)}=\nse \left(\p(x) \frac{e^\eps-1}{e^\eps+1}+\frac{1}{e^\eps+1}\right),
\]
which has a similar form as~\eqref{eqn:pre-image}. Hence setting $\beta=1/(e^\eps+1)$ and $\gamma=(e^\eps-1)/(e^\eps+1)$, we can follow the steps in the communication constrained setting and obtain with probability at least $9/10$, $d_{TV}(\hp, \p)\le \alpha$ using
\[
    n 
    = O\Paren{ \beta\frac{s^2\log \frac {\ab}s}{\alpha^2\gamma^2} }
    = O\Paren{\frac{s^2 \max \{ \log (\ab/s), 1\} }{\alpha^2 \eps^2}},
\]
when $\eps = O(1)$.